\documentclass[format=acmsmall, review=false]{acmart}

\usepackage{acm-ec-17}
\usepackage{booktabs} 

\usepackage{amsmath}
\usepackage{amsfonts}
\usepackage{amssymb}
\usepackage{graphicx}
\usepackage{caption}
\usepackage{amsthm}
\usepackage{placeins}
\usepackage[]{url}

\usepackage[T1]{fontenc}
\usepackage[utf8]{inputenc}
\usepackage{placeins}

\pdfoutput=1

\begin{document}
\title{The Price of Indivisibility in Cake Cutting}
\titlenote{This work was done while the first author was an intern under the supervision of the second author at Conduent Labs, Bangalore, India.}

\author{Eshwar Ram Arunachaleswaran}
\affiliation{%
 \institution{Indian Institute of Science, Bangalore}
 \department{Computer Science and Automation}
}

\author{Ragavendran Gopalakrishnan}
\affiliation{%
  \institution{Cornell University}
  \department{School of Civil & Environmental Engineering}
  \city{Ithaca}
  \state{NY}
  \postcode{14853}
  \country{USA}
}

\begin{abstract}
We consider the problem of envy-free cake cutting, which is the distribution of a continuous heterogeneous resource among self interested players such that nobody prefers what somebody else receives to what they get. Existing work has focused on two distinct classes of solutions to this problem - allocations which give each player a continuous piece of cake and allocations which give each player arbitrarily many disjoint pieces of cake. Our aim is to investigate allocations between these two extremes by parameterizing the maximum number of disjoint pieces each player may receive. We characterize the Price of Indivisibility (POI) as the gain achieved in social welfare (utilitarian and egalitarian), by moving from allocations which give each player a continuous piece of cake to allocations that may give each player up to $k$ disjoint pieces of cake. Our results contain bounds for the Price of Indivisibility for utilitarian as well as egalitarian social welfare, and for envy-free cake cutting as well as cake cutting without any fairness constraints.
\end{abstract}

\maketitle

\section{Introduction}

Cake cutting is a way to abstract the problem of dividing a continuous heterogeneous resource among multiple agents or players in a fair manner. The underlying idea is that different parts of a cake such as chocolate icing, nuts or berries represent different values to the people who are dividing the cake analogous to how the valuation of the agents vary heterogeneously over a continuous resource such as land or server time slots. There are multiple definitions of ``fairness'' in the context of cake cutting. One such definition of a fair allocation is one that is proportional, where each of the $N$ players get at least $\frac{1}{N}$ of their value of the cake. A stronger notion of fairness (which implies proportionality when the entire cake is allocated) is that of an \textit{envy-free} allocation, which guarantees that no player prefer some other player's allocated portion of cake. Throughout this paper, we adopt our notion of fairness to be that of envy-freeness. Applications of cake cutting include land division, partitioning advertising time slots, and the fair allocation of computing resources. The MARA survey of \citet{Chevaleyre06issuesin} provides a detailed description of the application areas.

Allocations in cake cutting, and thus the algorithms used to obtain them, are partitioned into two categories:
\begin{enumerate}
    \item \textit{Contiguous Allocations}, where every player gets a single piece of cake, and
    \item \textit{Non-Contiguous Allocations}, where each player may get arbitrarily many disjoint pieces of cake.
\end{enumerate}
In real-life applications such as land division or the partitioning of server time slots, Contiguous Allocations are preferred to Non-Contiguous Allocations because the latter give no guarantee over the number or sizes of the disjoint chunks of cake. For example, it is impractical to give infinitely (or even exponentially) many infinitesimally small crumbs of land to an agent in a land division problem. On the other hand, the social welfare of Non-Contiguous Allocations can be significantly larger than that of Contiguous Allocations. In addition, from a computational point of view, Non-Contiguous allocations are preferable because there is no finite protocol for the problem of finding envy-free Non-Contiguous Allocations \cite{Complexity}. Furthermore, there is recent work \cite{Aziz} that gives a finite bounded-time protocol for finding an envy-free Non-Contiguous Allocation.\footnote{This work also bounds the number of cuts (and hence the total number of pieces of cake), but this bound is of an impractically high order ($n^{n^{n^{n^{n^{n}}}}}$) to satisfy the concerns about players getting meaninglessly small crumbs of cake.}

Within the cake cutting literature, there is work that focuses on optimality (with respect to utilitarian welfare) \cite{Cohler_aaai11}, bounded time complexity \cite{Kurokawa:2013:CCB:2891460.2891537}, and truthfulness \cite{RePEc:eee:gamebe:v:77:y:2013:i:1:p:284-297}. These results are limited to cake cutting instances with restricted valutation functions (piecewise uniform or piecewise linear). The literature can also be classified according to the admissible allocations, that is, Contiguous and Non-Contiguous Allocations.

The goal of this paper is to analyze allocations that lie between the two extremes of Contiguous and Non-Contiguous Allocations, by defining ``$k$-Contiguous Allocations'', as follows.

\begin{definition}
Given an instance $i$ of a cake cutting problem, $Allocations(k)(i)$ refers to the set of Non-Contiguous Allocations of $i$ such that each agent is allocated at most $k$ disjoint pieces of cake, and $Allocations_f(k)(i)$ refers to the subset of $Allocations(k)(i)$ that only include envy-free allocations.
\end{definition}

In particular, we wish to characterize how the optimal social welfare (both utilitarian and egalitarian) grows with $k$. It is natural to expect that relaxing the constraint of allocating a single piece to each player to allowing up to $k$ pieces per player should improve social welfare, regardless of any additional fairness constraints. Therefore, we first study the effect of this relaxation for the plain resource allocation problem (without any fairness constraints), before moving on to the problem of envy-free cake cutting. This additional step allows us to analyze the impact of imposing envy-freeness. In summary, we characterize the growth of both utilitarian and egalitarian social welfare, with and without the constraint of envy-freeness. We call this growth the ``price of indivisibility'', and define four variants of this measure, one for each of the above four cases. Throughout, let $I(n)$ denote the set of all cake cutting instances with $n$ agents.

\begin{definition}
The maximum increase in utilitarian social welfare when each agent can be allocated up to $k$ disjoint pieces of cake is defined as
\begin{equation}
POIU(n,k) = \max_{i \in I(n)}\dfrac{\max_{a\in Allocations(k)(i)}\mbox{Utilitarian Welfare }(a)}{\max_{a\in Allocations(1)(i)}\mbox{Utilitarian Welfare }(a)}
\end{equation}
\end{definition}

\begin{definition}
The maximum increase in utilitarian social welfare when each agent can be allocated up to $k$ disjoint pieces of cake subject to envy-freeness is defined as
\begin{equation}
POIU_f(n,k) = \max_{i \in I(n)}\dfrac{\max_{a\in Allocations_f(k)(i)}\mbox{Utilitarian Welfare }(a)}{\max_{a\in Allocations_f(1)(i)}\mbox{Utilitarian Welfare }(a)}
\end{equation}
\end{definition}

The measures $POIE(n,k)$ and $POIE_f(n,k)$ for the egalitarian social welfare can be analogously defined. These definitions are inspired by the concept of ``Price of Fairness'', defined as the ratio of the optimal social welfare without the constraint of fairness to the optimal social welfare with the constraint of fairness. It was first studied for Non-Contiguous Allocations by \citet{Caragiannis2012} considering utilitarian social welfare, and for three notions of fairness (envy-freeness, proportionality, and equitability). \citet{Aumann:2015:EFD:2810066.2781776} then studied the price of fairness for Contiguous Allocations, considering both utilitarian and egalitarian social welfare, and for the same three notions of fairness.

\subsection{Our Results}

Our primary results comprise either exact values or strong lower bounds for each of the four variants of POI. In particular, we show the following:


\[ POIU(2,k) = 2 - \dfrac{1}{k} \]

\[ POIE(n,k) \begin{cases}
      \ge k & k\leq n \\
      = n & k > n
   \end{cases} \]

\[ POIU_{f}(n,k) \begin{cases}
      \ge k & k\leq n \\
      = n & k > n
   \end{cases} \]

\[ POIE_{f}(n,k) \begin{cases}
      \ge k & k\leq n \\
      = n & k > n
   \end{cases} \]

The first result above is an exact computation of the gain in utilitarian social welfare when there are only two players. We observe that this quantity is \textit{concave} in $k$; therefore, it grows swiftly for small values of $k$, and slows down for larger values, inching towards $2$ (the maximum value) when $k\rightarrow\infty$. In contrast, the quantities $POIU_{f}$, $POIE$, and $POIE_{f}$ grow linearly with $k$ until $k$ approaches $n$, at which point they stop growing. In the process of proving bounds for $POIU_{f}$ and $POIE_{f}$, we have proved a secondary result regarding the existence of envy-free contiguous allocations with certain desirable properties for a limited set of cake cutting instances, obtained by suitably adapting the approach of \citet{eqbm}, where the authors exhibit a game whose pure Nash equilibria are envy-free allocations.

The following are conjectures which we believe to be true:
\begin{enumerate}
    \item $POIU_{f}(n,k) = POIE_{f} = POIE_{f} = O(k)\quad \forall k \in \{1,2,3...,n\}.$
    \item $POIE(n,k) < POIE_{f}(n,k)\quad \forall k \in \{1,2,3...,n\}.$
    \item $POIU(n,k)$ is a concave \textit{concave} in $k$ for all values of $n$.
\end{enumerate}

\subsection{Some Related Work}

Our work attempts to characterize the growth of social and egalitarian welfare while increasing the maximum number of pieces allocated to a player, by drawing upon the definitions of the ``price of fairness'' from \citet{Caragiannis2012} and \citet{Aumann:2015:EFD:2810066.2781776}. Our concept, the ``price of indivisibility'' is analogous to the price of fairness, wherein we replace the fairness constraint with the constraint on the maximum number of continuous pieces allocated to a player.

The work of \citet{Caragiannis_towardsmore} deals with the aforementioned problem of many small crumbs in Non-Contiguous Allocations by introducing a new class of valuation functions. These ``Piecewise Uniform with Minimum Length'' (PUML) valuations ensure that each disjoint piece of cake allocated to any agent is  always longer than a specified minimum length, thus implicitly bounding the number of pieces each player receives. This work restricts the valuations to being piecewise uniform and presents algorithms for optimal approximately proportional, envy-free allocations.

Our approach can also be viewed as a workaround to the limits on efficiency placed by the constraints of envy-freeness and connectedness, as presented by \citet{Aumann:2015:EFD:2810066.2781776}. The work of \citet{Arzi2011} attempts to do something similar by characterizing the improvement in utilitarian and egalitarian social welfare of envy-free Contiguous Allocations obtained by permitting throwing away a bounded portion of the cake. It is worth noting that throwing away a part of the cake can only improve the social welfare (dubbed as the ``dumping effect'') of envy-free Contiguous Allocations. In other words, there is no ``dumping effect'' for envy-free Non-Contiguous Allocations, because each disjoint dumped piece of cake could be treated as a separate cake to be allocated in an envy-free manner (Proposition 1 in \cite{Arzi2011}). Our constraint relaxation (allowing up to $k$ continuous pieces per player, instead of just one) is, perhaps, comparable to theirs (relaxing geometric constraints), in the sense that these are different concessions that can be made in order to improve social welfare. While the potential egalitarian welfare gain in both approaches turns out to be $O(n)$, it is worth noting that the utilitarian welfare gain can be as much as $O(n)$ in our approach (when $k\geq n$), whereas it is limited to at most $O(\sqrt{n})$ in that of \citet{Arzi2011}.\newline
Finally, we note the work of \citet{2017arXiv171009477N} regarding fair (i.e. envy-free) divisions of cake with each player receiving multiple pieces. This work proves existential results for the division of a cake in a manner that allows a lower bounded number of players to prefer mutually disjoint sets of $k$ pieces.

\section{Model and Notations}

We consider a  rectangular cake C$^{*}$, represented as the interval $[0,1]$, which is cut into disjoint intervals by cuts parallel to the left edge of the cake. The set of players is denoted by $N=\{1,2,3..,n\}$. Every player $i \in N$ has a valuation function $V_{i}$ : $ 2^{C^{*}} \rightarrow [0,1]$ . This valuation function has the following properties -

\begin{enumerate}
	\item Non-Negativity : $ V_{i}([a,b]) \geq 0$
	\item Additivity : For disjoint intervals $I_{1}$ and $I_{2}$, $ V_{i}(I_{1}  \cup I_{2} ) = V_{i}(I_{1}) + V_{i}(I_{2}) $'
	\item Normalization : $ V_{i} ([0,1]) = 1 $
	\item Well Defined : The valuation functions are well defined in nature. They do not distinguish between open and closed intervals , i.e., there are no point values - $ V_{i}([a,a]) = 0 $
	\item Divisibility For every interval [a,b] and $ 0 < f \le 1$  , there exists a sub-interval [c,d] i.e. $ a \le c < d \le b$   such that $ V_{i}([c,d]) = f . V_{i}([a,b]) $
	
\end{enumerate}

Additionally, each of these valuation functions $V_{i}$ have an underlying continuous, non negative value density function $d_{i}$ such that for every interval I $\subset$ [0,1], $ V_{i}(I) = \int_{x \in I} d_{i}(x) dx $ . The continuity (and finiteness) of the density function is what guarantees that the valuation function is well defined throughout the cake.\newline
A portion of the cake is a union of disjoint intervals of cake. A portion is continuous if it is a single interval of the cake. We define an allocation to be $A ={A_{1},A_{2},....A_{n}}$ with all the portions being mutually disjoint. Player $i$ gets portion $A_{i}$.\newline
Each node has a utility for any allocation- $U_{i} = V_{i}(A_{i})$. The social welfare of an allocation is the sum of the utilities of all the nodes.
Thus, an envy-free allocation has the property that $V_{i}(A_{i}) \ge V_{i}(A_{j}) \forall i,j \in N$. \newline
For notational convenience, we define the ``gain'' - $\tau_{f}$ and $\tau$ (with and without fairness respectively) of any cake cutting instance i with n players and k pieces.
\begin{equation}
\tau_{f}(n,k,i) = \dfrac{\max_{a\in Allocations_f(k)(i)}\mbox{Utilitarian Welfare }(a)}{\max_{a\in Allocations_f(1)(i)}\mbox{Utilitarian Welfare }(a)}
\end{equation}
\begin{equation}
\tau(n,k,i) = \dfrac{\max_{a\in Allocations(k)(i)}\mbox{Utilitarian Welfare }(a)}{\max_{a\in Allocations(1)(i)}\mbox{Utilitarian Welfare }(a)}
\end{equation}
Thus:
\begin{equation}POIU(n,k) = \max_{i \in I(n)} \tau (n,k,i) \end{equation}
\begin{equation}POIU_{f}(n,k) = \max_{i \in I(n)} \tau_{f}(n,k,i)\end{equation}

\subsection{A Restriction on Cake Cutting Instances}

We will restrict our analysis to a certain set of cake cutting instances $X$. This set $X$ is the union of sets $X(1), X(2), X(3),$... where $X(j)$ has the following definition.

\begin{definition}
Every cake cutting instance in $X(j)$ can be constructed in the following manner. Every player $i \in N$ ``owns'' j rectangular pieces $i_{1}$ to $i_{j}$ such that $V_{l}(i_{m})=0 \ for \ l \neq i \ and \ 1 \le m \le j$. The cake is any permutation of the these n.j rectangular pieces of cake.
\end{definition}
Note that $X(i) \subset X(j)$ if $i<j$.\newline
There are two reasons why we analyze cake cutting instances only in the set X.

\begin{enumerate}
    \item We are looking for cake cutting instances in which the optimum social welfare increases considerably when moving from Allocations(1) to Allocations(k). Intuitively, this gain is maximized when the cake is composed of distinct rectangular pieces without overlapping values, thus making it easier to pick k pieces of ``good'' value for each player (because nobody else values them). This intuition is formalized in a later section.
    \item Finding optimal allocations and optimal envy-free allocations is significantly easier for instances in X as compared to more general cake cutting instances.
\end{enumerate}

\section{The Price of Indivisibility}

\subsection{POIU without Fairness}

We will restrict our analysis to cake cutting instances in the set $X(k)$ and later prove that searching in this set will find the instance that maximizes $\tau$. Note that the numerator of $\tau$ - optimum utilitarian welfare in Allocations(k)- is always $n$ for such instances - as every player can get all the pieces owned by him (which are the only pieces he values).

\begin{lemma}
Any allocation maximizing utilitarian welfare for instances in $X(k)$ can be restricted to making cuts only on the boundaries between the $n.k$ pieces.
\end{lemma}
\begin{proof}
Assume that there exists some allocation with a cut C not on the boundaries of the $n.k$ pieces. Consequently, this cut divides one of the $n.k$ piece (which belongs to player P) into two halves with each half going to a different players P1 and P2. Now, there are two cases, when player P is one of these two players, and when P is not. In the latter case, since neither P1 and P2 have any value for the piece, we can move the cut C to either end of the piece without changing utilitarian welfare. In the former case, we can in fact improve utilitarian welfare by giving the entire piece to player P. Thus, we can restrict our search to only allocations that maximize utilitarian welfare to only those which have cuts on the boundaries. Clearly, this proof also works for multiple cuts within a piece, the two cases being whether player P is among the set of players receiving the divided pieces or not.
\end{proof}
This lemma is of particular utility to us because we can now treat finding an allocation with optimum utilitarian welfare as a combinatorial optimization problem.\newline
We begin with the result for the two player case.
\subsubsection{The Two Player Case}

\begin{theorem}
We present a tight bound on POIU for the 2 player case : \newline
$ \max_{i \in X(k)} \tau (2,k,i) = 2- \dfrac{1}{k} $
\end{theorem}
	
\begin{proof}
Our proof is constructive in nature. Firstly, the numerator of $\tau$ is fixed to $2$ because of the assumption of searching for allocations in $X(k)$. Therefore, we only need to think about minimizing the denominator, that is, the optimum utilitarian welfare within Contiguous Allocations. Since we are dealing with two players here, that is equivalent to finding a single cut point on the cake that maximizes the utilitarian welfare of the resulting allocation. Note that for any cut point, deciding who gets the left side and who the right is obvious from the point of view of optimizing utilitarian welfare, the left piece is given to the player who values it more and as a consequence, the right piece is valued more by the other player (conveniently for us).
\newline
Additionally, because of Lemma 1, we only need to find which one among the $2k-1$ potential cut points maximizes the social welfare. Let us label these cut points - $ C = [C_{1},C_{2},....,C_{2k-1}]$. Thus, for a cake cutting instance in X(k):
\[ \max_{a\in Allocations(1)}\mbox{Utilitarian Welfare }(a) = \max_{c \in C} \{max\{V_{1}([0,c]) + V_{2}([c,1]), V_{2}([0,c]) + V_{1}([c,1]) \}\}\]
Intuitively, it seems that reducing the utilitarian welfare at one of the cut points comes at the cost of increasing it elsewhere. We will formalize this in the next lemma where we prove that the sum of the utilitarian welfare across all the cuts is at least 2k. Armed with this lemma and the expression for optimum utilitarian welfare above, we can say that:  $\max_{a\in Allocations(1)}\mbox{Utilitarian Welfare }(a) \ge \dfrac{2k}{2k-1}$. Since the sum of the $2k-1$ possible values of utilitarian welfare is fixed to be at least $2k$ and the optimum utilitarian welfare is the maximum of the $2k-1$ values, the smallest possible value of optimum utilitarian welfare occurs when all the $2k-1$ possible values are equal and add up to exactly $2k$. To complete the proof, we construct a cake cutting instance in $X(k)$ where $\max_{a\in Allocations(1)}\mbox{Utilitarian Welfare }(a)$ is exactly $\dfrac{2k}{2k-1}$. We follow this up with a lemma that proves that the sum of social welfares across the cut points in C is lower bounded by $2k$.\newline We begin with an instance (Fig 1) for the case of $k=2$.
\begin{figure}[h!]
	\centering
	\includegraphics{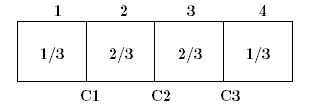}
	\caption{}
	\label{fig:img20170710113647867}
\end{figure}

It is obvious that the utilitarian welfare ($UW$) across each of the cuts C1, C2 and C3 is equal to $1 + \dfrac{1}{3}$. To explain the construction this instance, let the corner pieces have value $a$ and $b$ and the center piece have value $1-a$ and $1-b$ - for players 1 and 2 respectively. Let us assume that $a,b < \dfrac{1}{2} $. Thus,
\[ UW_{C1} = 1 + a\]
\[ UW_{C2} = 2 - (a+b)\]
\[ UW_{C3} = 1+ b \]
Now, we want to have $UW_{C_{1}} = UW_{C_{2}} = UW_{C_{3}}$, this gives us $a = b = \dfrac{1}{3}$ . Onto the more general case with $k$ pieces, let the value of the corner piece (of say, player 1 - although the reasoning applies to both players) be $x$. Therefore, $UW_{C_{1}} = 1+x$. Let the value of the piece [$C_{1},C_{2}$] be $y$ (to player 2). Consider the cut at C2, we want $UW_{C_{1}} = UW_{C_{2}}$ - thus the allocation has to switch sides (by which we mean that the left side piece and the right side pieces are swapped between the players) from C1 (else $UW_{C_{1}} > UW_{C_{2}}$ because of the drop in value for the player who gets the right side piece). Therefore,

\[UW_{C_{2}} = 1 + (y-x)  = UW_{C_{1}} = 1+x \]
\[ \implies y=2x\]

Let the value of [$C_{2},C_{3}$] be z (to player 1). Again, the allocation has to switch sides from C2 (else $UW_{C_{2}} > UW_{C_{3}}$ because of the drop in value for the player who retains the left side piece). Thus,

\[UW_{C_{3}} = 1 + (z-x) = UW_{C_{1}} = 1+x \]
\[ \implies z=2x\]

We extend this result using induction, that is, all pieces but the corner piece have value $2x$ if the corner piece has value $x$. Consider any piece $[C_{i},C_{i+1}]$ belonging to player P (either 1 or 2) of value $a$. We know that $SW_{C_{i}} = 1+x$, and that the allocation has to swap sides at cut $C_{i+1}$. Thus,
\[UW_{C_{i+1}} = (2 - UW_{C_{i}}) + a = UW_{C_{1}} = 1+x\]
\[\implies 1 + (a-x) = 1 +x \]
\[\implies a = 2x \]

Finally, we normalize the values of the pieces, each player has one (corner) piece of value $x$ and $(k-1)$ pieces of value $2x$.
\[\implies x + (k-1)*2x = 1 \]
\[\implies x = \dfrac{1}{2k-1} \]
\[\implies \max_{a\in Allocations(1)}\mbox{Utilitarian Welfare }(a)= 1 + \dfrac{1}{2k-1}\]

This is exactly the value that is our lower bound. Consequently,

\[\max_{i \in X(k)} \tau (2,k,i)  = \dfrac{2}{1+\dfrac{1}{2k-1}}\]
\[\implies \max_{i \in X(k)} \tau (2,k,i)  = 2 - \dfrac{1}{k}\]
\end{proof}

We now prove the lower bound on the sum of social welfares across the cuts in C to complete the proof of theorem 3.2.
\begin{lemma}
	\[ \sum_{1}^{2k-1} UW_{C_{i}} \ge 2k\]
\end{lemma}

\begin{proof}
	We will prove this inequality by rewriting the sum and grouping terms. Let $UW_{C_{i}} = L_{i} + R_{i} $ where $L_{i}$ and $R_{i}$ represents the respective utilities of the players who receive the left and right side pieces for cut $C_{i}$ . We rewrite the sum as:
	\[ \sum_{1}^{2k-1} UW_{C_{i}} = (\sum_{1}^{2k-2} L_{i} + R_{i+1}) + R_{1} + L_{2k-1} \]
	
	We know that $L_{2k-1}$ and $ R_{1}$ are both equal to $1$, because the corner piece goes to the player to whom it belongs and the rest of the cake (which contains the entire value of the other player) goes to the other player.
	
	Consider each of the $2k-2$ terms in the other sum, we claim that :
	\[ \forall i \  in \{1,2,3,...,2k-2\} : L_{i} + R_{i+1} \ge 1 \]
	
	To prove this claim, consider the cuts at $C_{i}$ and $C_{i+1}$ - without loss of generality, we assume that player 1 has the left side piece for the allocation at cut $C_{i}$ there are two cases - whether the piece $[C_{i},C_{i+1}]$ belongs to player 1 or 2 , each with two sub-cases whether the allocation switches sides or not from $C_{i}$ to $C_{i+1}$. Let the left piece at $C_{i}$ (i.e.$[0,C_{i}]$) have value $x$ to player 1 and $y$ to player 2. Because of our assumption, $x>y$. Also, depending on the case, $[C_{i},C_{i+1}]$ has value $x'$ to player 1 or value $y'$ to player 2.
	
	There are two cases, with two sub-cases each.
	\begin{description}
	
	\item[Case 1:] In Fig 2, Piece $[C_{i},C_{i+1}]$ belongs to Player 1 \newline
	
	\begin{figure}[h!]
		\centering
		\includegraphics{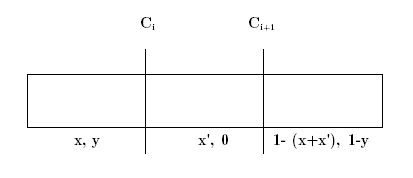}
		\caption{}
		\label{fig:case-1}
	\end{figure}
	
	\begin{description}
	\item[Sub-Case 1:] Allocation does not flip
	
	\[ L_{i} + R_{i+1} = x + (1-y) = 1+ (x-y) >1 as x>y\]
	
	\item[Sub-Case 2]: Allocation Flips \newline
	
	This is not possible as $(x+x')>y$ as $x>y$ and all valuations are non-negative.
	\end{description}
	\item[Case 2:] In Fig 3, Piece $[C_{i},C_{i+1}]$ belongs to Player 2 \newline
	\begin{figure}[h!]
		\centering
		\includegraphics{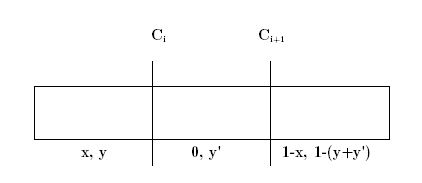}
		\caption{}
		\label{fig:case-2}
	\end{figure}
	\item[Sub-Case 1:] Allocation Flips
	\[\implies x < y+y' \]
	\[  L_{i} + R_{i+1} = x + (1-x) =1 \]
	\item[Sub-Case 2:] Allocation does not flip
	\[\implies x> y+y' \]
	\[ L_{i} + R_{i+1} = x + 1 - (y+y') = 1 + (x - (y+y')) \ge 1\]

    \end{description}
	
	Thus, each of the $2k-2$ terms is at least $1$.
	
	\[\implies \sum_{1}^{2k-1} UW_{C_{i}} \ge (2k-2) + 1 +1 \]
	\[\implies \sum_{1}^{2k-1} UW_{C_{i}} \ge 2k \]
	
	Hence Proved.
\end{proof}

So far, we have restricted our search for the $POIU(2,k)$ to the set of cake cutting instances X(k). However, the price of indivisibility is defined over the set of all cake cutting instances $I(2)$ for two players.

\begin{equation}
POIU(n,k) = \max_{i \in I(n)}\dfrac{\max_{a\in Allocations(k)(i)}\mbox{Utilitarian Welfare }(a)}{\max_{a\in Allocations(1)(i)}\mbox{Utilitarian Welfare }(a)}
\end{equation}

Basically, we are setting the numerator to be $n$ and trying to find the minimum possible value of the denominator.

So far we have the result that:

\[  \max_{i \in X(k)} \tau (2,k,i) = 2- \dfrac{1}{k}  \]
 and a corresponding instance $I'$ which shows exactly this gain (with the corresponding values of numerator and denominator). Now, if we were to shift to a generic cake cutting instance $I^{*}$ ($I^{*} \notin X(k)$, no point looking in $X(k)$ anymore) from I' - there is a resultant decrease in the numerator as well as the denominator - $l_{n}\ and\ l_{d}$. Note that $0 \le l_{n} \le 1$ (as the optimum utilitarian welfare is always greater than 1) and $0 \le l_{d} \le \dfrac{1}{2k-1}$   Thus:
 \[ \tau(I^{*}) = \dfrac{2-l_{n}}{\dfrac{2k}{2k-1}-l_{d}} \]
 We want to prove that
 \[ \tau(I^{*}) < 2 - \dfrac{1}{k} \]
 which comes down to proving that:
 \[ \dfrac{l_{n}}{l_{d}} > \dfrac{2k-1}{k} \]

 \begin{lemma}
 	\[\dfrac{l_{n}}{l_{d}} \ge 2k-1 \]
 	
 	This lemma states that whatever loss is caused (from $\dfrac{2k}{2k-1}$) in the denominator, the optimum utilitarian welfare with connected pieces - that loss is magnified by $2k-1$ in the numerator, i.e optimum social welfare with each player getting $k$ connected pieces.
 \end{lemma}

 \begin{proof}
 	The intuition is that instance $I'$ has $2k-1$ clearly defined cut points (the cut set C) such that the social welfare is equal across all these cut points (and their sum has a fixed lower bound). Additionally, the allocation flips at each cut point. The reduction $l_{d}$ must be effected across all these cut points as the one with the greatest social welfare is picked. The same $2k-1$ points are the exact set of cut points for an optimum allocation with $k$ pieces in the instance $I'$.\newline
 	First we define a variables to capture the current loss in the numerator - $l'_{n,i}$. We know $2k-1$ cuts are made in total when each player gets $k$ connected pieces (such that the social welfare is optimal) - we name this cut set C'. We also know that the pieces to player 1 and 2 alternate. As part of our proof, we make the cuts one by one from left to right and cumulatively measure the loss in the numerator given the constraint that the loss in the denominator has to be $l_{d}$. There is a direct one to one mapping between the cut set C in instance $I'$ and the cut set C in instance $I^{*}$ in our proof where $C'_{i}$ maps to $C_{i}$. We define $UW_{i}$ and $UW'_{i}$ as the respective utilitarian welfares in $I^{*}$ and I' after the leftmost i cuts have been made.
 	\[ l'_{n,i} = UW'_{i} - UW_{i} \]
 	
 	Clearly $l'_{n,2k-1} = l_{n}$.
 	
 	We prove by induction that \[l'_{n,i} \ge i(l_{d}) \]

 	Consider the base case-  that is, only the first cut is made:
 	\[ l'_{n,1} \ge l_{d} \] is trivially true, otherwise we have an allocation with continuous pieces with the loss in the denominator being less than $l_{d}$.
 	
 	Now, let the claim be true for some $i$. Now, consider what happens at $i+1$, firstly the allocation with connected pieces flips at this point - so what was earlier the loss until cut $C'_{i}$ - (at least) $l_{d}$ in the denominator now becomes the denominator gain (with respect to instance $I'$ and cut $C_{i+1}$ - i.e $\dfrac{2k}{2k-1}$) because of the flipping, thus we need to have a denominator loss of (at least) $l_{d}$ in the new allocation with cut $C'_{i+1}$ added to maintain the denominator constraint and thus this results in completely ``reversing the loss'' in the numerator and adding an extra loss of $l_{d}$ . Thus, this takes up the total loss in the numerator up to $(i+1)l_{d}$. Hence, our induction argument is complete. Thus,
 	\[ l_{n} \ge l'_{n,2k-1} \ge (2k-1)l_{d}\]
 	
 	\[\implies \dfrac{l_{n}}{l_{d}} \ge 2k-1 \ge \dfrac{2k-1}{k}  \]
 	
 	This completes our proof for the value of $POIU$ for the $n=2$ case.

 \end{proof}

\begin{theorem}
\[ POIU(2,k) = 2 - \dfrac{1}{k} \]
\end{theorem}

\subsection{POIE without Fairness}

To analyze the price of indivisibility with respect to egalitarian welfare, we again choose cake cutting instances in $X(k)$ for analysis because we easily obtain strong lower bounds for $POIE$ by searching in this space. Note that the optimum envy-free egalitarian welfare in Allocations(k)- is always $1$ for such instances - as every player can get all the pieces owned by him (which are the only pieces he values).\newline
We begin by proving two lower bounds for egalitarian welfare (EW) with connected pieces.

\begin{lemma}
	$\forall i \in I(n) :\mbox{Optimal EW}(n,k,i) \ge \dfrac{1}{n}$, where I represents the set of all cake cutting instances
\end{lemma}

\begin{proof}
	The proof is constructive in nature. We know that a proportional allocation with continuous pieces always exists - as one can always be generated by the protocol of \citet{Dubins-Spanier} . A proportional allocation always guarantees a value of at least $\dfrac{1}{n}$ to each player and thus proves the lemma.
\end{proof}

\begin{lemma}
		$\forall i \in X(k) :\mbox{Optimal EW}(n,k,i) \ge \dfrac{1}{k}$
\end{lemma}

\begin{proof}
	Again, we offer a constructive proof, consider an allocation in which each player gets (among other things) the most valuable piece that belongs to him. The rest of the cake is divided arbitrarily between the players who hold the ``most valuable pieces'' enclosing each unallocated section of the cake. Clearly, each player has utility at least $\dfrac{1}{k}$ because by the pigeon-hole principle, the value of the most valuable piece is at least $\dfrac{1}{k}$. Thus, our result is proved.
\end{proof}

Combining these two results, the real lower bound i.e., the larger one, is $\dfrac{1}{k}$ when k<n and $\dfrac{1}{n}$ for $k \le n$. i.e.:

\[ \min_{\forall i \in X(k)} \max_{a\in Allocations(1)(i)}\mbox{Egalitarian Welfare }(a) \ge  \begin{cases}
 \dfrac{1}{k} & k<n \\
 \dfrac{1}{n} & k \ge n
\end{cases} \]

Note that the optimum egalitarian welfare with $k$ pieces being given to each player is $1$ because we are restricted to allocations in $X(k)$.

\begin{theorem}
	\[POIE(n,k) \begin{cases}
	\ge k & k <n  \\
	= n & k \ge n
	\end{cases} \]
\end{theorem}

\begin{proof}
	The proof is constructive in nature. We show an instance for $ k \le n$ which has the smallest possible optimum egalitarian welfare possibles. The instance has k identical blocks of n pieces. In each block, each piece belongs to a distinct player and all the blocks have the pieces arranged in the same order (of the players to which they belong) . Each of the $n$ pieces have value $\dfrac{1}{k}$. We claim that the optimum egalitarian welfare of this arrangement is $\frac{1}{k}$ - trivially achieved by giving each player the piece that belongs to him from the first block and giving the rest of the cake to the player who owns the adjoining piece. Consider the possibility that optimum egalitarian welfare is greater than $\dfrac{1}{k}$. Each player would then strictly need to get more than one piece that belongs to him as each piece has value exactly $\dfrac{1}{k}$. However, two pieces that belong to the same player are always separated by $n-1$ other pieces- hence each player must get n+1 pieces. However this is not possible for $k \le n$ because the total number of pieces allot ed must be $n(n+1)$ which is strictly greater than the $n.k$ pieces available. Thus, egalitarian welfare cannot be larger than the $\dfrac{1}{k}$ already achieved. This takes care of $k<n$ and $k=n$. Consider the instance for $k=n$, we have constructed an instance with optimum egalitarian welfare $\dfrac{1}{n}$- this is the smallest possible optimum egalitarian welfare $k>n$ as per the bounds we have proved. Thus, this instance can be directly used for all $k>n$ as well to prove the theorem.
\end{proof}

\subsection{POIU with Fairness}

We analyze instances in $X(k)$ to obtain some strong lower bounds for $POIU_{f}$. Note that the numerator of $\tau_{f}$ - optimum envy-free utilitarian welfare in Allocations(k)- is always $n$ for such instances - as every player can get all the pieces owned by him (which are the only pieces he values).

\subsubsection{The Two Player Case}

We show a construction in Fig 4 with two players with two pieces that gives us close to the maximum possible gain of $2$. We know the gain cannot be better than $2$ because every envy free allocation is also proportional when it covers the entire cake. Hence, every player has utility at least $\dfrac{1}{n}$ and thus the optimum envy free social welfare with connected pieces is at least $1$.

\begin{figure}[h]
	\centering
	\includegraphics{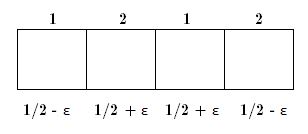}
	\caption{}
	\label{fig:2-player-fair}
\end{figure}

 The optimal envy free allocation comes with a cut down the exact centre of the cake (the boundary separating the left two pieces from the right two pieces). It is clear that no player can gain any more without causing a greater loss to the other player while maintaining envy-freeness. Thus, $\max_{a\in Allocations_f(1)}\mbox{Utilitarian Welfare }(a) = 1 + 2\epsilon$. Thus the gain and hence the $POIU_{f}(2,k) \rightarrow 2\ as\ \epsilon \rightarrow 0$. We generalize this result to $k$ pieces because the same construction can be used for the $k$ pieces as well because we have already hit the maximum possible gain.
\subsubsection{Some Lower Bounds}

We state some useful observations in the form of lemmas before proving our results.

\begin{lemma}
	The utilitarian welfare of any envy free allocation is at least 1.
	\[ \implies \max_{a\in Allocations_f(1)}\mbox{Utilitarian Welfare }(a)\ge 1 \]
\end{lemma}

\begin{proof}
Every EF allocation of the entire cake is also proportional, hence each player has utility at least $\dfrac{1}{n}$
\end{proof}

The implication is that $POIU_{f}(n,k) \le n$.

\begin{lemma}
	No optimal envy free allocation gives any player only a part of a piece ``owned'' by him.
\end{lemma}
\begin{proof}
	An allocation that gives the entire piece to the player (while the rest is unchanged) has greater utilitarian welfare and does not affect the envy-freeness property.
\end{proof}
\begin{lemma}
	Every player gets at least complete one piece owned by him in an envy-free Contiguous Allocation.
\end{lemma}
\begin{proof}
	Every EF allocation of the entire cake is also proportional, so each player has utility at least $\dfrac{1}{n}$. So every player has to get at least part of his own cake. From the previous lemma, we know he gets entire pieces only. Hence Proved.
\end{proof}
\begin{lemma}
	None of the $n.k$ pieces of cake are divided among more than two players in an optimal envy free Contiguous Allocation.
\end{lemma}
\begin{proof}
	This is a direct consequence of the previous two lemmas. We know the player to whom a piece belongs will never just get a part of a cake. Consider a division of the piece among three or more players, none of whom ``own'' that piece - Clearly at least one of them ends up with an utility of $0$ and no piece of his own - thus this is not possible.
\end{proof}

\begin{theorem}
	Optimal envy-free allocations in the $2$-pieces case ($k=2$) always have an utilitarian welfare of at least $\dfrac{n}{2}$.
\end{theorem}

\begin{proof}
	Each player has two pieces that belong to him, we know from lemma 3.11 that each player gets at least one of these two pieces. Consider an allocation where each player gets (among other things) the higher value that belongs to him. The rest of the cake can be divided arbitrarily between the two players who hold the neighbouring high value pieces belonging to them. Clearly this allocation is envy free, because no player having his high value piece can envy someone having the other piece belonging to him. Since the higher value piece has value at least $\dfrac{1}{2}$, any optimal, continuous envy free allocation has value at least $\dfrac{n}{2}$.
\end{proof}

Now, we show a construction of a cake cutting instance T with $n$ pieces owned by each player (thus $T \in X(n)$) that has a gain approaching the maximum possible gain of $n$. This can be used for the $k$ pieces case where $k \ge n$ .

\begin{theorem}
$POIU_{f}(n,k) \rightarrow n$ when $ k \ge n$.
\end{theorem}

\begin{proof}
Our instance consists of $n$ blocks of $n$ pieces , each block has pieces belonging to all the players. Each player has $n-1$ ``symmetric'' pieces of value $\dfrac{1}{n} - \epsilon$ and one ``asymmetric'' piece of value $\dfrac{1}{n} + (n-1)\epsilon$. Each block of $n$ pieces has exactly one asymmetric piece and that belongs to player number $n-i$ for block number $i$ (counting from left to right). Fig 5 is a depiction of this instance with the asymmetric pieces having stripes.

\begin{figure}[h]
\centering
\includegraphics[width=1\linewidth]{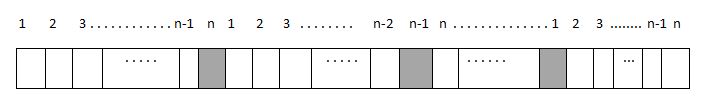}
\caption{}
\label{fig:nn-player-fair}
\end{figure}

We claim that the optimal envy-free allocation with continuous pieces is each player getting (among other things) exactly one piece that belongs to them with this being the asymmetric piece. All the symmetric pieces are divided arbitrarily between the players whose asymmetric pieces they are sandwiched between. This arrangement is clearly envy-free. To prove it is optimal, consider an allocation in which a player receives more than one piece belonging to him (among other things obviously)- one of them has to be a symmetric piece. Now, this symmetric piece comes at the cost of some player missing out on his asymmetric piece, which is a clear drop in optimality (trading an asymmetric piece for a symmetric piece for different players reduces social welfare). Additionally, to keep such an allocation envy free, this other player must receive at least two symmetric pieces again causing another player to lose his asymmetric piece in exchange creating a continuous chain of loss in optimality (we are not even sure such a chain ends in an envy free allocation). Hence, $\max_{a\in Allocations_f(1)}\mbox{Utilitarian Welfare }(a) = 1 + n(n-1)\epsilon$ which approaches $1$ as $\epsilon \rightarrow 0$. This instance can also be used to generate the $POIU$ for the $k$ pieces case where $k \ge n$ as we have already demonstrated the maximum possible gain. Thus,
\[ POIU_{f}(n,k) \rightarrow n\ as\ \epsilon \rightarrow 0\ when\ k \ge n \]

\end{proof}

Next up, we deal with the cases when $k < n$.

\begin{lemma}
 In an optimal EF allocation with continuous pieces (i.e an EF Contiguous Allocation), each player gets at most two partial pieces belonging to some other player.
\end{lemma}

The implication is that each optimal EF continuous allocation has at most 2n divided pieces.

\begin{theorem}
	$\forall i \in X(k): \max_{a\in Allocations_f(1)(i)}\mbox{Utilitarian Welfare }(a)\ge \dfrac{n}{k+2}$ where $k < n$.
\end{theorem}

\begin{proof}
	Each player $i$ gets at least one piece that belongs to him. Let his utility be $u_{i}$. We know that the rest of the pieces that belong to him can be at most $k-1$ in number, let these pieces be divided among $(k-1)+s_{i}$ players. Assuming that as many players as possible divide these pieces to minimize the minimum utility of each player (may not always be possible, but we are only establishing a lower bound), hence we set $0 \le s_{i} \le k-1$ - the upper bound is because each of these pieces can be divided among at most two players. Clearly,
	\[\sum_{i=1}^{n} s_{i} \le 2n \]
	
	Due to the envy-freeness condition:
	\[ \forall i in {1,2,...,n} : u_{i} \ge \dfrac{1-u_{i}}{k-1 + s_{i}} \implies u_{i} \ge \dfrac{1}{k+s_{i}} \]
	\[\implies Optimum\ Social\ Welfare_{1-piece} = \sum_{i=1}^{n} u_{i} \ge \sum_{i=1}^{n} \dfrac{1}{k+s_{i}} \]
	
	We treat $k+s_{i}$ as a new variable and using the AM-HM inequality conclude that:
	\newline $\forall i \in X(k): \max_{a\in Allocations_f(1)(i)}\mbox{Utilitarian Welfare }(a)\ge \dfrac{n}{k+2}$ where $k < n$
	
\end{proof}

\subsubsection{The Maximum Value Piece Idea}
However, we can prove a stronger real lower bound of $\dfrac{n}{k}$ for $\forall i \in X(k): \max_{a\in Allocations_f(1)(i)}\mbox{Utilitarian Welfare }(a)$. Theorem 3.11 provides the intuition for the idea that - for an optimal EF Contiguous Allocation, each player gets at least the value of the most valuable piece owned by him. If there were an equivalent lemma to lemma 3.1 (proving that cuts need to made only on the boundaries between the $n.k$ ``owned'' pieces) for optimal envy-free allocations for cake cutting instances in $X(k)$, then establishing this maximum value piece idea would be straightforward. A simple proof by contradiction is as follows -if any player P does not get the value of his most valuable piece (MVP), then the hypothetical lemma analogous to lemma 3.1 would ensure that some other player received a piece of cake containing P's MVP, thus causing P to envy this player. However, such a lemma clearly cannot be proved in a similar manner to lemma 3.1 . This is because, in certain envy-free allocations, some of the $n.k$ pieces may be divided between two players who do not own that piece to ensure envy-freeness (as the player who owns that piece might envy some other player who gets the entire piece), even if there is no consequence in terms of utilitarian welfare in doing so.
\newline

Consider the following protocol inspired by the Dubins-Spanier Protocol. A knife moves from left to right on the cake starting from the left extreme. Any player who values the cake to the left of the knife equal to the value of his most valuable piece can call and claim that piece. Note that we adhere to the idea of lemma 3.10 and let each player delay the call so that he gets the entirety of the piece that belongs to him at the right end of his piece (although this is not strictly necessary for what we are proving, this allows for greater utilitarian welfare. We claim that this protocol ensures that every player gets at least the value of his most valuable piece. The proof rests on the fact that there are n such most valuable pieces. In each round of the procedure,we claim that each player who has not received a portion of cake has his most valuable piece unallocated currently - otherwise, he would have already called for it and received that piece as no other player would call with the knife in the middle or end of that piece. Hence, after each of the first n-1 rounds, the remainder of the cake always has ``enough'' value for the active players thus proving our claim. However, we cannot guarantee that the allocation resulting from this protocol will be envy free, as it is always possible that a player who has already received a piece of cake might envy a later piece called for by another player.
\newline

\begin{lemma}
An envy free contiguous allocation exists such that the allocation gives each player at least the value of his most valuable piece (of the k pieces owned by him).
\end{lemma}

Note that this lemma does not talk about optimal envy free contiguous allocations. However, it helps in bounding the quantity - $\forall i \in X(k): max_{a\in Allocations_f(1)(i)}\mbox{Utilitarian Welfare }(a)$, and also helps with a later result regarding $POIE_{f}$.

\begin{proof}
The proof is constructive in nature, and comes up with a protocol that generates an envy free allocation with the desired properties.
\newline

We look to the work of \citet{eqbm} to prove the existence of an envy free allocation with the maximum value piece property. This work is done in the setting where every player knows all about the other player's preferences. Firstly, we convert this protocol to the Dubins-Spanier protocol based game in \citet{eqbm} with threshold strategies. We then adapt the results of the paper for our particular class of cake cutting instances by modifying the game. These modifications are required because \citet{eqbm} restrict the value density functions to be strictly positive which rules out cake cutting instances such as those in the set X where the players are allowed to have zero value over parts of the cake. In contrast, we work with non-negative value density functions which allow players to have zero value over pieces of cake.\newline To summarize the protocol of \citet{eqbm}- the protocol is in the form of a n-round game G where each player's strategy is in the form of $n$ ``thresholds''. Player $i$ had strategy $S_{i} = (t_{i}^{1},t_{i}^{2},....t_{i}^{n})$, meaning that player $i$ will call whenever the value of the cake to left of the knife is $t_{i}^{r}$ in round $r$ (and receives the cake to the left of the knife) if he has not already received a piece. \newline
We make some modifications to this protocol to create a new game G' for players to divide cake cutting instances in the set X. These modifications ensure that the game G' only induces cuts on the boundaries of the n.m ``owned'' pieces. \newline
TThe main result of \citet{eqbm} is that a pure Nash equilibrium in their game (which always exists for some deterministic tie breaking rule) results in an allocation which is envy free. We will show in the next lemma that the results of \citet{eqbm} hold for G' as well despite the modifications we make to the protocol. We have already shown a simple strategy that any player can use to get at least the value V' of his most valuable piece regardless of the actions of the other players in this game. In each round, this player P only calls if the value of the cake to the left of the knife is equal to V' - this immediately guarantees that this player will get V' as no other player will  cut in the middle of player P's most valuable piece. Thus, the Nash equilibrium that induces an envy free allocation must assure each player of at least the value of his most valuable piece or have its Nash equilibrium property violated by the existence of a more profitable strategy to deviate to. This proves the maximum value piece lemma.
\end{proof}

We provide a brief explanation of the proof of the primary result of \citet{eqbm}. The protocol of \cite{eqbm} has an additional feature - a static tie-breaking rule, a permutation $\pi$ of players - {1,2,3..,n}, to choose one among multiple players who may call at the same time. However, for cake cutting instances in X with our modified protocol, which results in a subset of equilibria of the original game, there is no need for a tie-breaking rule because players do not call simultaneously, and hence there will always be a pure Nash equilibrium (the existence of a pure Nash equilibrium in the more general setting depends upon the tie breaking permutation, always being possible for certain permutations). We will list some observations about pure Nash equilibria (NE) in the original game of \citet{eqbm}

\begin{description}
    \item[Property 1:] In every pure NE, the entire cake is always allocated.
    \item[Property 2:] In every round except the last round of a game where the players have strategies corresponding to a pure NE, two players call simultaneously.
\end{description}

\begin{lemma}[\citet{eqbm}]
Every pure NE in the game G induces an envy free contiguous allocation of cake.
\end{lemma}

\begin{proof}
This is a proof by contradiction. Let A be an allocation induced by a pure NE. Let player $i$, who gets his piece $A_{i}$ in round $r_{i}$ envy player j, who gets his piece $A_{j}$ in round $r_{j}$. In the case where $r_{j}<r_{i}$, player $i$ can deviate profitably by mimicking player $j$ to get piece $A_{j}$ which he prefers (hence, proving that this strategy profile is not a NE). The case $r_{j}>r_{i}$ is a little more complicated. Player $i$ can still get $A_{j}$ by mimicking player $j$, this is made possible by the ``substitute'' caller at every round from $r_{i}$ to $r{j}-1$ who makes sure the same cuts are made.
\end{proof}

Now, onto pure NE in our game G' for cake cutting instances in X. Our problem with the game G is that players will delay their call until the next player is about to call (or exactly when they call in case of a favourable tie-breaking rule). This means that the cuts on the cake will not be only on the boundaries between the $n.k$ ``owned'' pieces. Therefore we artificially force conditions favourable to us. We enforce the following two rules:

\begin{enumerate}
    \item Cuts are only to be made at the boundaries between the $n.k$ ``owned'' pieces.
    \item The players must delay their call until the knife reaches the beginning of the owned piece P at the end which the next caller C will call (with piece P belonging to player C). Thus, the two "simultaneous callers" of the game G are separated by a piece owned by the ``substitute'' caller. This restriction enforces a particular class of equilibria that is favourable to us.
\end{enumerate}

\begin{lemma}
Every pure NE in the game G' induces an envy free contiguous allocation of cake.
\end{lemma}

\begin{proof}
Although property 1 of pure NE in G holds for any pure NE A', property 2 clearly does not hold- because the valuations of only one player changes at any point of time when a knife moves over the cake. That leaves us with the troublesome case of $r_{j}>r_{i}$ when player i, who gets his piece $A'_{i}$ in round $r_{i}$ envy player j, who gets his piece $A'_{j}$ in round $r_{j}$. We claim that it is a profitable strategy for player i to mimic player j, and will result in player i receiving a piece of value  $V_{i}(A'_{j})$ even if player i does not receive the piece $A_{j}$ itself. Although two players do not call simultaneously like in the game G, the ``substitute'' caller calls one piece later, without affecting the dynamics of the other players because this piece is of zero value to anybody else, as far as they are concerned, that part of the cake does not exist for anybody but the ``substitute'' caller. Thus, although the actual cuts will be shifted to the right in some pattern, player i can still get a piece of value $V_{i}(A_{j})$ by mimicking player j.
\end{proof}

Note that G' is a specific instance in the class of games that can be defined by applying the rules of the game G to cake cutting instances in X. This is because G does not specify the course of action for a player when there are portions of cake that are of no value to the player who is about to call. This particular definition of G' is designed to simulate the ``Two players call simultaneously'' property of the game G while preserving the advantageous property of restricting cuts to the boundaries between ``owned'' pieces (which makes possible the alternative strategy which guarantees each player the value of his most valuable piece regardless of the actions of the other players).
\newline
To go with lemma 3.17, we come up with a simple construction that takes us to the lower bound of $\dfrac{n}{k}$ - this cake cutting instance T' in $X(k)$ has $k$ ``inner'' players and $n-k$ ``outer'' players. The inner player's pieces are arranged in a pattern similar to construction T in fig 5 - with $k$ blocks of $k$ pieces, one belonging to each player, arranged in ascending order of player number, with the $i^{th}$ block having the asymmetric (infinitesimally larger value than the other pieces owned by player i). The outer players just have k pieces of value $\dfrac{1}{k}$ each, and they are arranged in k blocks of (n-k) players (in no particular order). The blocks of the inner and outer players alternate (there are k blocks of each, thus this is possible). We claim that an optimal envy free allocation has each player receiving (among other pieces) only one piece that belongs to them, the asymmetric piece in the inner players' case and any piece in the outer players' case. The proof follows a similar line of the proof involving the instance T, namely that two pieces belonging to any player always have an asymmetric piece in between them. Thus
for this instance - $max_{a\in Allocations_f(1)}\mbox{Utilitarian Welfare }(a) \rightarrow \dfrac{n}{k}$. Thus, the gain approaches $k$. Hence, we have given a construction for the lower bound and obtained the desired result.

\begin{theorem}

\[ POIU_{f}(n,k) \begin{cases}
      \ge k & k\leq n \\
      = n & k > n
   \end{cases} \]

\end{theorem}

\subsection{POIE with Fairness}

 The instance T has the maximum possible value of $POIE_{f}$ of n (when $k \ge n$). Note that the maximum value piece lemma offers the following result regarding optimal egalitarian welfare in $X(k)$ when $k < n$. Also note that the optimum egalitarian welfare in Allocations$_f$(k)- is always $1$ for instances in $X(k)$ - as every player can get all the pieces owned by him (which are the only pieces he values).

\begin{lemma}
$\forall i \  \in X(k): \max_{a\in Allocations_f(1)(i)}\mbox{Egalitarian Welfare }(a) \ge \dfrac{1}{k}$
\end{lemma}

\begin{proof}
Each player can always get at least the value of his MVP, which is $\dfrac{1}{k}$ as per the pigeon hole principle.
\end{proof}

Thus, the following result is obtained:
 \begin{theorem}
     \[ POIE_{f}(n,k) \begin{cases}
      \ge k & k\leq n \\
      = n & k > n
   \end{cases} \]

 \end{theorem}
\section{Discussion}

We have presented certain results regarding the price of indivisibility by focusing on cake cutting instances in the set X. Interestingly, the price of indivisibility for utilitarian welfare ($POIU$) grows significantly faster for envy-free cake cutting as compared to just allocation of heterogeneous continuous resources. In contrast, the $POIE$ (for egalitarian welfare) is comparable with and without the constraint of envy-freeness. However, we have only compared the utilitarian or egalitarian welfare of (specifically) optimal allocations in Allocations($1$) and Allocations($k$). Thus, any algorithms for envy-free cake cutting along the lines of the work of \citet{RePEc:eee:gamebe:v:77:y:2013:i:1:p:284-297}, \citet{Cohler_aaai11} and \citet{Kurokawa:2013:CCB:2891460.2891537} must be careful to keep this distinction in mind. Merely finding an Envy-Free allocation in Allocations($k$) is not guaranteed to improve efficiency, there must be some structural guarantee for efficiency, for instance, that the envy-free allocation in Allocations($k$) that is found by the algorithm is more efficient than the envy-free allocation found by the same algorithm within Allocations($k-1$). The complexity of such algorithms is also an interesting area of future work. Since practical applications such as land division or advertising slot partition may demand small values of $k$ as compared to $n$, parameterized complexity and FPT algorithms are possible avenues of exploration for this problem.

\begin{acks}
	
We acknowledge the support of Conduent Labs India for funding the research that lead to this work.

\end{acks}

\bibliographystyle{ACM-Reference-Format}
\bibliography{paper_current_material}

\end{document}